\newif\ifJSC
\JSCfalse
\documentclass[sigconf, screen]{acmart}
\ifJSC
\usepackage[hyphens]{url}
\fi
\usepackage{todonotes}
\usepackage{algorithm}
\usepackage{algorithmic}
\usepackage{amsmath}
\usepackage{amsthm}
\usepackage{comment}
\usepackage{xspace}
\usepackage{mathtools}
\usepackage{subcaption}
\newtheorem{definition}{Definition}
\newtheorem{example}{Example}
\newtheorem{remark}{Remark}
\newtheorem{hypothesis}{Hypothesis}
\newtheorem{proposition}{Proposition}
\ifJSC
\journal{J.~Symbolic Computation}
\fi
\newcommand{\lexl}{lex-least\xspace}
\newcommand{\Lexl}{Lex-least\xspace}
\def\semival{semi-valuation\xspace}
\def\seminu{\nu^{\prime}}

\newcommand{\ProjBM}{P_{BM}} 
\newcommand{\LiftBM}{L_{BM}} 

\def\bottom{\perp}
\def\CC{\mathbb{C}}
\def\NN{\mathbb{N}}
\def\RR{\mathbb{R}}
\def\ZZ{\mathbb{Z}}
\def\coeff{\mathop{\rm coeff}\nolimits}
\def\cont{\mathop{\rm cont}\nolimits}
\def\disc{\mathop{\rm disc}\nolimits}
\newcommand\ldcf{\operatorname{\rm lc}}
\newcommand\trcf{\operatorname{\rm tc}}
\def\lc{\mathop{\rm lc}\nolimits}
\def\res{\mathop{\rm res}\nolimits}
\def\tc{\mathop{\rm tc}\nolimits}

\newcommand\PL{\operatorname{\rm PL}}

\newcommand{\Test}{\mathbf{T}}

\setcopyright{acmcopyright}
\copyrightyear{2023}
\acmYear{2023}
\acmDOI{DOI Goes Here}
\acmConference[ISSAC'23]{ISSAC'23: International Symposium on Symbolic and Algebraic Computation}{Dates}{City, Country}
\acmBooktitle{ISSAC'23: International Symposium on Symbolic and Algebraic Computation, Dates, City, Country}
\acmPrice{15.00}
\acmISBN{ISBN Goes Here}

\usepackage{color}
\copyrightyear{2023}
\acmYear{2023}
\setcopyright{rightsretained}
\acmConference[ISSAC 2023]{International Symposium on Symbolic and Algebraic Computation 2023}{July 24--27, 2023}{Tromsø, Norway}
\acmBooktitle{International Symposium on Symbolic and Algebraic Computation 2023 (ISSAC 2023), July 24--27, 2023, Tromsø, Norway}\acmDOI{10.1145/3597066.3597090}
\acmISBN{979-8-4007-0039-2/23/07}
\begin{document}




\title{Lazard-style CAD and Equational Constraints}


\author{James H. Davenport}
\affiliation{%
  \institution{University of Bath}
  \country{United Kingdom}
}
\orcid{0000-0002-3982-7545}
\email{masjhd@bath.ac.uk}

\author{Akshar Nair}
\affiliation{%
  \institution{University of Bath}
  \country{United Kingdom}
}
\orcid{0000-0001-7379-1868}
\email{akshar.nair@gmail.com}

\author{Gregory Sankaran}
\affiliation{%
  \institution{University of Bath}
  \country{United Kingdom}
}
\orcid{0000-0002-5846-6490}
\email{g.k.sankaran@bath.ac.uk}

\author{Ali K. Uncu}
\affiliation{%
  \institution{RICAM, Austrian Academy of Sciences\\ \& University of Bath}
  \country{Austria \& United Kingdom}
}
\orcid{0000-0001-5631-6424}
\email{akuncu@ricam.oeaw.ac.at}

\thanks{We acknowledge UKRI EPSRC for their constant support. The second author's thesis was supported by EPSRC grant EP/N509589/1. The first and the last authors are partially supported by EPSRC grant EP/T015713/1. The last author also thanks the partial support of Austrian Science Fund FWF grant P34501-N.
We are grateful to Chris Brown for explanations of \cite{BrownMcCallum2020a}.}
\bibliographystyle{ACM-Reference-Format}

\begin{abstract}
McCallum-style Cylindrical Algebra Decomposition (CAD) is a major improvement on the original Collins version, and has had many subsequent advances, notably for total or partial equational constraints. But it suffers from a problem with nullification. The recently-justified Lazard-style CAD does not have this problem. However, transporting the equational constraints work to Lazard-style does reintroduce nullification issues. This paper explains the problem, and the solutions to it, based on the second author's Ph.D. thesis and the Brown--McCallum improvement to Lazard.
\par
With a single equational constraint, we can gain the same improvements in Lazard-style as in McCallum-style CAD . Moreover, our approach does not fail where McCallum would due to nullification. Unsurprisingly, it does not achieve the same level of improvement as it does in the non-nullified cases. We also consider the case of multiple equational constraints.
\end{abstract}

\ccsdesc[500]{Computing methodologies~Symbolic and algebraic manipulation}

\keywords{Cylindrical algebraic decomposition, Lazard projection and lifting, Equational constraints} 




\maketitle
\section{Introduction}
\emph{Cylindrical Algebraic Decomposition} (CAD) was introduced by Collins in 1975 \cite{Collins1975}.
There have been many improvements since \cite{Collins1975} introduced the projection-lifting paradigm. McCallum \cite{McCallum1985b} made a major one, but this had the drawback that it could not be applied if a polynomial nullified. There have been many developments to \cite{McCallum1985b}, often concerned with \emph{equational constraints}, i.e. when the semi-algebraic set lies in a proper sub-variety. \cite{McCallumetal2019a} justified Lazard's idea in \cite{Lazard1994}, using different forms of projection and lifting to avoid the nullification problem.

This paper is a first step in transplanting the improvements to \cite{McCallum1985b} to the Lazard setting. More specifically, we base ourselves on the Brown--McCallum \cite{BrownMcCallum2020a} version of the Lazard approach.
\par
Previous papers \cite{Nairetal2019b,Nairetal2020a} were based on \cite{McCallumetal2019a}.
The second author's thesis \cite{Nair2021b} was largely based on \cite{McCallumetal2019a}, with some updates for \cite{BrownMcCallum2020a}.
Here we integrate \cite{BrownMcCallum2020a} throughout.

\section{Background and Notation}
\subsection{CAD by Projection-Lifting}
There are various ways of computing cylindrical algebraic decompositions, e.g. via Triangular Decomposition \cite{Chenetal2009d}, but this paper is cast in the context of the Projection-Lifting paradigm introduced by Collins \cite{Collins1975}: see Figure \ref{fig:PL}. Let $S\subset\ZZ[x_1,\ldots,x_n]$ be the set of polynomials of interest, and $\phi(S)$ the property we wish our decomposition to have. Generally $\phi(S)$ will be a variant of ``each cell is sign-invariant for the polynomials in $S$'', as it is that property that enables quantifier elimination, Collins' original motivation. 
\begin{figure}[h]
$$
\begin{array}{clclc}
\RR^n&S&&D_n&\phi_X(S)\cr
\strut&\downarrow P_X&&\uparrow L_X\cr
\RR^{n-1}&P_X(S)&&D_{n-1}&\phi_X(P_X(S))\cr
&\downarrow P_X&&\uparrow L_X\cr
\vdots&\vdots&&\vdots&\vdots\cr
\RR^2&P_X^{n-2}(S)&&D_2&\phi_X(P_X^{n-2}(S))\cr
&\downarrow P_X&&\uparrow L_X\cr
\RR^1&P_X^{n-1}(S)&\rightarrow\hbox{isolate roots}\rightarrow&D_1&\phi_X(P_X^{n-1}(S))
\end{array}
$$
where $X$ indicates the precise variant of Projection/Lifting, and $D_i$ is the intermediate decomposition of $\RR^i$.
\caption{Projection/Lifting Paradigm\label{fig:PL}}
	\Description{The Projection/Lifting Paradigm}
\end{figure}

The original Collins algorithm for Cylindrical Algebraic Decomposition \cite{Collins1975} have an expensive (compared to its successors) projection operator $P_C$, and $\phi_C$ is indeed sign-invariance. The lifting operation $L_C$ is simple: above each cell $C\in D_i$ there is a cylinder $C\times\RR\subset\RR^{i+1}$ which is sliced by the real branches of $p_k=0$ for the polynomials $p_k\in P^{n-i-1}(S)$ and $P_C$ ensures that these branches are delineable over $C$, i.e. do not intersect and are simple surfaces above the whole of $C$. The relevant cells of $D_{i+1}$ are then each branch, known as ``sections'', and the ``sectors'' between adjacent branches, including the sector ``$-\infty<x_{i+1}<$ all sections`` and the sector ``all sections $<x_{i+1}<\infty$''.

\subsection{Curtains}
Most improvements to Collins' algorithm can encounter difficulties when some polynomials nullify, i.e. vanish over some subset of $\RR^k=\langle x_1,\ldots,x_k\rangle$.
The term curtain for the varieties where polynomials nullify was introduced in~\cite{Nairetal2020a}:
\def\foo{\cite[Definition 43]{Nair2021b}}
\begin{definition}[\foo]
        A variety $C\subseteq\RR^n$ is called a curtain if, whenever
        $(x,x_n)\in C$, then  $(x,y)\in C$ for all
        $y\in\RR$.
\end{definition}

$C$ is a curtain if it is a union of fibres of $\RR^n \to \RR^{n-1}$.
\def\foo{\cite[Definition 44]{Nair2021b}}
\begin{definition}[\foo]
        Suppose $f\in\RR[x_1,\ldots,x_n]$ and $S\subseteq\RR^{n-1}$. We say
	that $V_f$ (or $f$) has a \emph{curtain} at $S$ if for all $(\alpha_1,\ldots,\alpha_{n-1})\in S$,
	$y\in \RR$ we have $f(\alpha_1,\ldots,\alpha_{n-1},y)=0$. We call $S$ the \emph{foot} of the curtain.
	If the foot $S$ is a singleton, we call the curtain a \emph{point-curtain}.
\end{definition}

\def\foo{\textsc{cf} \cite[Definition 45]{Nair2021b}}
\begin{definition}[\foo]\label{def:implicit_explicit}
Suppose the polynomial $f\in\RR[x_1,\ldots,\allowbreak x_n]$ factorises as $f=gh$, where $g\in\RR[x_1,\ldots,x_{n-1}]$ and $g(\alpha_1,\ldots,\allowbreak \alpha_{n-1})=0$. Then the variety of $f$ is said to contain an {explicit} curtain whose foot is the zero set of $g$. 
	A curtain which does not contain (set-theoretically) an explicit curtain is said to be {implicit}, and a curtain which contains an explicit curtain but is not an explicit curtain (because $h$ itself has curtains) is said to be {mixed}. 
\end{definition}

\begin{example}[Types of Curtain]\label{eg:EC_IMC} ${ }$

	\begin{itemize}
		\item Explicit Curtain: $f(x,y,z)=xy^2-y^2-xz+z=(x-1)(y^2-z)$, curtain at $(1,0)$. The curtain can be seen in Figure \ref{fig:cur_EC} as the sheet given by $x=1$.

		\item Implicit Curtain: $f(x,y,z)=x^2+yz$, curtain at $(0,y,0)$.
The blue line represents this curtain in Figure~\ref{fig:cur_IC}.
 
	\end{itemize}
\end{example}

\begin{figure}[h!]
	\vskip-12pt 
	\begin{subfigure}{0.2\textwidth}
		\centering
		\includegraphics[scale=0.17]{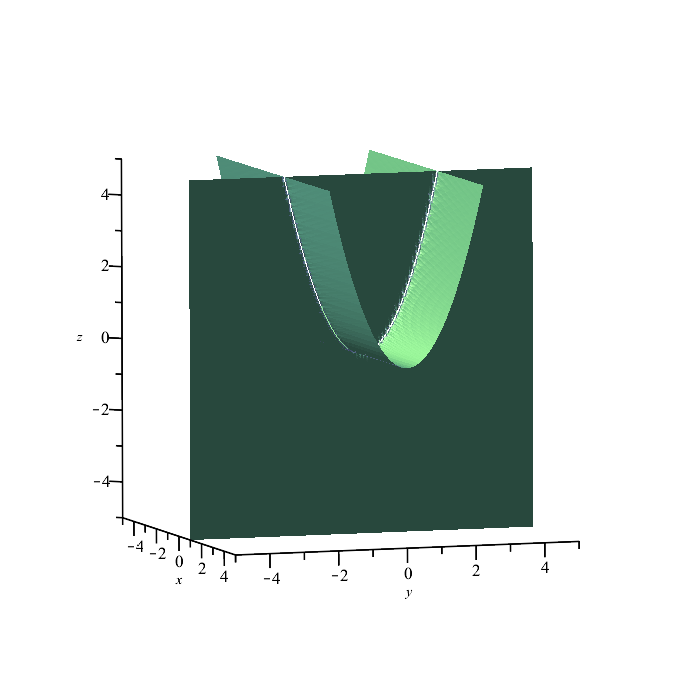}
		\caption{Surface with an Explicit Curtain}
		\label{fig:cur_EC}
	\end{subfigure}
	\hfill
	\begin{subfigure}{0.2\textwidth}
		\includegraphics[scale=0.17]{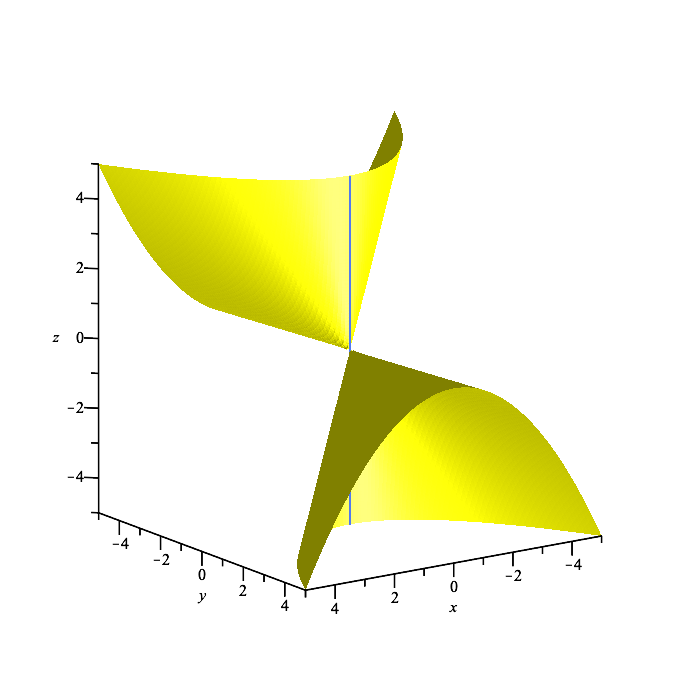}
		\caption{Surface with an Implicit Curtain}
		\label{fig:cur_IC}
	\end{subfigure}
	\vskip-12pt 
	\caption{Different types of curtains }
	\Description{The images corresponding to the two examples of curtains}
	\label{fig:three graphs}
\end{figure}

\subsection{McCallum's Improvement to Collins}
\def\foo{\cite[p. 46]{McCallum1985b}}
\begin{definition}[\foo]\label{def:McC}
Let $A$ be a set of polynomials in $\ZZ[x_1,\ldots,\allowbreak x_n]$, $\cont(A)$ the contents of the elements of $A$ with respect to $x_n$, and $B$ be a square-free basis\footnote{\cite{McCallum1985b} uses the finest square-free basis, but this is for reasons of efficiency.} for $A$. Let $\coeff(B)$ be the set of all coefficients of the members of $B$ (viewed as polynomials in $x_n$), $\disc(B)$ be the set of discriminants of elements of $B$ (with respect to $x_n$) and $\res(B)$ the set of all resultants  (with respect to $x_n$) of pairs of elements of $B$. McCallum's projection operator $P_M$ is defined as
\begin{equation}\label{eq:PM}
	P_M(A):=\cont(A)\cup\coeff(B)\cup\disc(B)\cup\res(B).
\end{equation}
\end{definition}
Using this definition, McCallum~\cite{McCallum1985b} improves on the  Collins' original~\cite{Collins1975}. The projection $P_M$ is distinctly smaller than $P_C$, and the corresponding $\phi_M$ is order-invariance, i.e. if a polynomial vanishes on a cell $C\in D_i$ it does so to constant order on that cell. The corresponding lifting operator $L_M$ is identical to $L_C$, except that, if a polynomial $q$ vanishes on a zero-dimensional cell, it is replaced by the corresponding \emph{delineating polynomial\/} \cite[Definition 1]{McCallum1998}. However, $L_M$ may fail (giving an error message, not the wrong result) if relevant polynomials nullify (are identically zero) over a cell of $\RR^k$ of positive dimension.

\subsection{Brown's Improvement to McCallum}
Brown~\cite{Brown2001b} introduced a useful improvement on  \cite{McCallum1985b}.
\begin{definition}\label{def:PB}
We use the same notation as Definition \ref{def:McC}, and  $\lc$ for the leading coefficient. Brown's projection is defined as 
\begin{equation}\label{eq:PB}
	P_B(A):=\cont(A)\cup\lc(B)\cup\disc(B)\cup\res(B).
\end{equation}
\end{definition}
The corresponding lifting $L_B$ has several (necessary) enhancements over $L_M$ but the details are largely irrelevant here: the relevant one is that, to the CAD we are building over, we add those points above which some projection factor nullifies (if it nullifies on a set of positive dimension, then the system is not well-oriented). For any non-trivial set $A$ of polynomials, $P_B(A)\subsetneq P_M(A)$, the difference being the non-leading coefficients of $B$. This does not affect the $\mathcal{O}$-asymptotics, which are driven by the $\res(B)$ term, but is very effective in practice. \cite[Tables 1,2]{Brown2001b}  consider an example $A$ in six variables, where $|P_B^4(A)|=18$, while $|P_M^4(A)|=129$.  $|P_B^5(A)|=141$, while  $P_M^5(A)$ could not be computed, but probably  $|P_M^5(A)|>8000$.

\subsection{Lex-least valuation and Lazard lifting}

Lazard~\cite{Lazard1994} introduced a novel approach to 
CAD to bypass nullification, which was later made rigorous by \cite{McCallumetal2019a}. This section contains the relevant definitions, as used in \cite{Nair2021b,Tonks2021a}.

\begin{definition}
Let $v,w\in\ZZ^{n}$. We say that $v=(v_1,\ldots,v_n)\geq_{lex}(w_1,\ldots,w_n)=w$
if and only if either $v=w$ or there exists an $i \leq n$ such that
$v_i>w_i$ and $v_k=w_k$ for all $k$ in the range $1 \leq k < i$.
\end{definition}
\def\foo{\cite[Definition 2.4]{McCallumetal2019a}}
\begin{definition}[\foo]\label{def:lexleast}
        Let $n\geq 1$ and suppose that $f\in \RR [x_1,\ldots,x_n]$ is non-zero
        and $\alpha=(\alpha_1 , \ldots, \alpha_n )\in \RR^n$. The \lexl
        valuation $\nu_\alpha (f)$ at $\alpha$ is the least (with respect to
        $\geq_{lex}$) element $v=(v_1,\ldots,v_n) \in \NN^n$ such that
        $f$ expanded about $\alpha$ has the term
        \[
        c(x_1-\alpha_1)^{v_1} \cdots (x_n-\alpha_n)^{v_n},
        \]
        where $c\neq 0$.
\end{definition}

Note that $\nu_\alpha(f)=(0,\ldots ,0)$ if and only if
$f(\alpha)\neq0$. We should note that the terminology we are using is slightly different than the previous ones. For example, \Lexl valuation is referred to as the Lazard
valuation in \cite{McCallumetal2019a}. We do it in an effort to be precise.

\Lexl valuation of $f$ at $\alpha\in\RR^n$ is calculated by first picking an order of variables $x_1\prec x_2\prec \dots\prec x_n$, then successively substituting in $\alpha_i$ in the given variable order. (The order of projections start with $x_n$ first, the order of valuation calculations starst with $x_1$ first.)  At each step of the successive substitutions $f(\alpha_1,\dots, \alpha_{i-1},x_i,\dots, x_n)$ we identify the largest positive integer $v_i$ such that \[(x_i-\alpha_i)^{v_i}\mid f(\alpha_1,\dots, \alpha_{i-1},x_i,\dots, x_n).\]  The $v_i$ is the component of the \Lexl valuation of $f$ at $\alpha$ that corresponds to $\alpha_i$. We continue the valuation calculations by first replacing $f(\alpha_1,\dots, \alpha_{i-1},x_i,\dots, x_n)$ with $f(\alpha_1,\dots, \alpha_{i-1},x_i,\dots, x_n)\times(x-\alpha_i)^{-v_i}$ and then substituting $\alpha_i$ in this new function.

\begin{example}
        If $n=1$ and $f(x_1)=x_1^3-2x_1^2+x_1= x_1(x_1-1)^2$, then $\nu_0(f)=1$,
        $\nu_1(f)=2$, and $\nu_x(f)=0$, for all $x\not= 0,1$. 
        Similarly, if $n=2$ (with the standard order $x_1\prec x_2$) and $f(x_1,x_2)=x_1(x_2-1)^2$, then
        $\nu_{(0,0)}(f)=(1,0)$, $\nu_{(2,1)}(f)=(0,2)$,
        $\nu_{(0,1)}(f)=(1,2)$, and $(0,0)$ otherwise.
\end{example}

The lex-least valuation is strongly dependent on the order of the variables, as the following example illustrates.
\begin{example}
        Let $f(x,y,z,w)=x^2+y^2z-2yz^2+zw$, $\alpha_1=(0,1,0,1)$ and $\alpha_2=(0,0,1,0)$. With respect to the ordering $x\prec y\prec z\prec w$ we get $\nu_{\alpha_1}(f)=(0,0,1,0)$ and $\nu_{\alpha_2}(f)=(0,0,0,1)$. With respect to the ordering $x\prec z\prec y \prec w$ we get $\nu_{\alpha_1}(f)=(0,1,0,0)$ and $\nu_{\alpha_2}(f)=(0,0,0,1)$. 
\end{example}
Note that in the case of $\alpha_2$, the ordering of variables does not change the valuation unlike the case for $\alpha_1$. Ordering of variables is essential and must be fixed when comparing valuations of points.
\def\foo{\cite[proposition 3.1]{McCallumetal2019a}}
\begin{proposition}[\foo]\label{prop:valuation}
         $\nu_\alpha$ is a valuation: that is, if $f$ and $g$
        are non-zero elements of $\RR[x_1,\ldots,x_n]$ and $\alpha\in \RR^n$,
        then $$\nu_\alpha(fg)=\nu_\alpha(f)+\nu_\alpha(g) \text{ and }
        \nu_\alpha(f+g)\geq_{lex}\min\{\nu_\alpha(f),\ \nu_\alpha(g)\}.$$
\end{proposition}
\def\foo{After \cite[Definition 2.6]{McCallumetal2019a}}
\begin{definition}[\foo]\label{def:residue}
	Let $n\geq 2$, and suppose that $f\in \RR[x_1,\ldots,x_n]$ is non-zero and that $\beta\in \RR^{n-1}$. The \lexl semi-valuation of $f$ on/above $\beta$, $\seminu_\beta(f)=(\nu_1,\ldots,\nu_{n-1})$, is defined as the vector of $n-1$ non-negative integers consisting of the valuation outcomes of $f$ at $\beta$, regarded as an element of $K[x_1,\ldots,x_{n-1}]$ where $K=\RR[x_n]$. 
	So $f=f_\beta(x_1-\beta_1)_{\nu_1}\cdots(x_{n-1}-\beta_{n-1})^{\nu_{n-1}}$, and we call $f_\beta$ the Lazard residue of $f$.
\end{definition}

The \lexl semi-valuation of $f$ is called the \textit{Lazard evaluation} of $f$ above $\beta$, denoted by $\nu_\beta(f)\ni\NN^{n-1}$ in~\cite{McCallumetal2019a}. The Lazard residue $f_\beta\in\RR[x_n]$ is called the \textit{Lazard evaluation} of $f$ at $\beta$ in~\cite{McCallumetal2019a}.

Note that~\cite{McCallumetal2019a} uses $\nu$ for two functions. When $\nu$ is used as a valuation \underline{at}, it is treated as a function from $\RR^{n}\times\RR[x_1,\ldots,x_n]\rightarrow\NN^n$. When $\nu$ is used as a valuation \underline{on}/\underline{above}, it is treated as a function from $\RR^{n-1}\times\RR[x_1,\ldots,x_n]\rightarrow\NN^{n-1}$. 
We felt that it is necessary to keep these two functions separate. The older notation also has the potential to cause confusion when reading $\nu_\alpha(f)$. The reader needs to check whether $\alpha\in\RR^{n}$ or $\alpha\in\RR^{n-1}$, or whether $\nu_\alpha(f)\in\NN^{n}$ or $\nu_\alpha(f)\in\NN^{n-1}$.


The 
\lexl \semival of $f$ on $\beta\in\RR^{n-1}$  must not  be confused
with the \lexl valuation at $\alpha\in\RR^n$, defined in
Definition~\ref{def:lexleast}. 
Notice
that if $b=(\beta,b_n)\in \RR^n$ then $\nu_b(f)=(\seminu_\beta(f),\nu_n)$
for some integer $\nu_n$: in other words, $\seminu_\beta(f)$ consists of the
first $n-1$ coordinates of the valuation of $f$ at any point above
$\beta$.

\Lexl \semival and Lazard residue are also dependent on variable ordering as demonstrated by the following example.
\begin{example}
        Let $f(x,y,z,w)=x^2+y^2z-2yz^2+zw$ and $\beta=(0,1,0)$ with ordering $x\prec y\prec z\prec w$ then  $\seminu_\beta(f)=(0,0,1)$ and $f_\beta=w+1$. If we change the ordering to $x\prec y\prec w\prec z$ then $\seminu_\beta(f)=(0,0,0)$ and $f_\beta=z-2z^2$.
\end{example}

\subsection{The Lazard Projection and Lazard Lifting}

Lazard's idea \cite{Lazard1994,McCallumetal2019a} is to instantiate Figure \ref{fig:PL} with $\phi$ being \lexl-invariance. The nullification problem of \cite{McCallum1985b} is bypassed by ``Lazard lifting'' $L_L$: when a polynomial $g \in P^{n-k}(S)$ nullifies on a cell $C$, its \lexl valuation is non-zero on $C$ (at the sample point $(\alpha_1,\ldots,\alpha_k)\in C$, and in the whole of $C$). The valuation is also constant over $C$ because $C$ is part of a \lexl-invariant cylindrical algebraic decomposition of $\RR^k$.  We replace $g$ by $g/(x_1-\alpha_1)^{v_1}\cdots (x_k-\alpha_k)^{v_{k}}$ when lifting (analogous to $L_C$, Collins' sign-invariant CAD lifting) over $C$ if this valuation is $(v_1,\ldots,v_k)$. 

For this to work, we need an appropriate projection operator $P_L$. This is defined in \cite{McCallumetal2019a} as follows.
\def\foo{\cite[Definition 2.1]{McCallumetal2019a}}
\begin{definition}[\foo]Let $A$ be a set of irreducible\footnote{\cite[Remark 2.2]{McCallumetal2019a} notes we could define this for square-free bases instead.\label{fnumber}} polynomials in $\ZZ[x_1,\ldots,x_k]$. The Lazard projection $P_L(A)$ is the union of these sets of polynomials in $\ZZ[x_1,\ldots,x_{k-1}]$. 
\begin{enumerate}
    \item All leading coefficients of the elements of $A$.
    \item All trailing coefficients of the elements of $A$.
    \item All discriminants of the elements of $A$.
    \item All resultants of pairs of distinct elements of $A$.
\end{enumerate}
\begin{equation}\label{eq:PL}
P_L(A):=\lc(A)\cup\tc(A)\cup\disc(A)\cup\res(A).
\end{equation}
\end{definition}
Though not explicit in \cite{McCallumetal2019a}, any contents of the polynomials (with respect to $x_k$) also carry forward into $P_L(A)$. We also note that $P_L\subsetneq P_M$ as the middle coefficients do not occur in $P_L$.
\par
Just as delineability is key to $L_C$ (and to $L_M$), we need a corresponding definition in this setting.
\def\foo{\cite[Definition 2.10]{McCallumetal2019a}}
\begin{definition}[\foo{}]
Let $f$ be a nonzero element of $\ZZ[x_1,\ldots, x_n]$ and $S$ a subset of $\RR^{n-1}$. We say that $f$ is Lazard delineable on $S$ if
\begin{enumerate}
\item
the \lexl valuation of $f$ on $\alpha$ is the same for each point $\alpha\in S$;
\item
there exist finitely many continuous functions $\theta_1<\cdots<\theta_k$ from $S$ to $\RR$, with $k\ge0$, such that, for all $\alpha\in S$, the set of real roots of $f_\alpha(x_n)$ is $\{\theta_1(\alpha),\ldots, \theta_k(\alpha)\}$ (empty when $k=0$); 
\item if $k >0$
there exist positive integers $m_1,\ldots m_k$ such that, for all $\alpha\in S$ and all $i$, $m_i$ is the multiplicity of $\theta_i(\alpha)$ as a root of $f_\alpha(x_n)$.
	\end{enumerate}
\end{definition}
If the \lexl valuation is zero, this is the same as ordinary delineability, otherwise it is delineability of the \lexl residue $f_\alpha$.

\subsection{The Brown--McCallum Improvement}
A key result on delineability of sections in the correctness proof for Brown--McCallum CAD is the following.
\def\foo{\cite[Theorem 3]{BrownMcCallum2020a}}
\begin{theorem}[\foo]\label{thm:projectiondelineable}
        Suppose that $f(x_1,\ldots,x_n)\in\RR[x_1,\ldots,\allowbreak x_n]$has positive degree $d$ in $x_n$ 
        and that $\disc_{x_n}(f)$ and $ \ldcf_{x_n}(f)$ are non-zero (as elements of $\RR[x_1,\ldots,x_{n-1}]$). Let $S$ be a connected analytic submanifold of $\RR^{n-1}$ in which $\disc_{x_n}(f)$ and $ \ldcf_{x_n}(f)$ are \lexl invariant, and at no point of which $f$ vanishes identically. Then $f$ is analytic delineable on $S$ and hence $f$ is \lexl invariant in every Lazard section and sector over $S$. 

\end{theorem}

This is less restrictive than \cite[Theorem 5.1]{McCallumetal2019a} for Lazard CAD where the delineability also requires the non-vanishing trailing coefficients and resultants. Conversely we have added the restriction ``at no point of which $f$ vanishes identically'', which is essentially introducing the concept of curtains.

For any non-trivial set $A$ of polynomials $P_B(A)\subsetneq P_L(A)$, the difference being the trailing coefficients, which appear in $P_L$ but not $P_B$. Can we eliminate them? Again, this would not affect the asymptotics, but might be useful in practice. The trailing coefficients are there for two reasons:
\begin{enumerate}
\item To identify curtains;
\item To ensure that the \lexl-invariant regions within the curtain form stacks.\label{Point2}
\end{enumerate}
If the feet of the curtains are isolated points, then reason \ref{Point2} is irrelevant, and Hypothesis~\ref{Hyp-BMcC} provides a different way of identifying curtains.

In \cite[p. 64]{McCallumetal2019a}, it is remarked that if 
$\lc(p)$ is nowhere vanishing, then $\tc(p)$ is not needed.
The argument in \cite{BrownMcCallum2020a} 
states that if $\lc(p)$ only vanishes at finitely many points, then $\tc(p)$ is not needed.
However, the argument and projection process are more complex.
\begin{hypothesis}[Brown--McCallum \cite{BrownMcCallum2020a}]\label{Hyp-BMcC}
	There is a useful and efficient test $\Test$ which, given $f\in\ZZ[x_1,\ldots,x_m]$, returns either a finite set $N$ of points in $\RR^{m-1}$ or $\bottom$. If $N$ is returned, then $N$ contains all the points $(\alpha_1,\alpha_2,\ldots\alpha_{m-1})$ at which $f(\alpha_1,\alpha_2,\ldots,\alpha_{m-1},x_m)$ vanishes identically. 
\end{hypothesis}
One way of satisfying the hypothesis is to compute a CAD of $\RR^{m-1}$ for the coefficients $a_i$ of $f=\sum a_i(x_1,\ldots,x_{m-1})x_m^i$. This is not as expensive at it seems, as it is only of the $a_i$, irrespective of all the other polynomials in the original problem. As in \cite{Wilsonetal2012a}, the authors would actually begin with a Gr\"obner basis of the $a_i$, and if this is zero-dimensional (in $\CC^{m-1}$) then $N$ is certainly finite.

The improved Lazard projection introduced in \cite{BrownMcCallum2020a}, denoted $\ProjBM$, takes, as well as a set of polynomials in $m$ variables, a set $\gamma_m$ of points in $\RR^m$ and returns, as well as a set of polynomials in $m-1$ variables, a set $\gamma_{m-1}$ of points in $\RR^{m-1}$.

\def\foo{\cite[Definition 4]{BrownMcCallum2020a}}
\begin{definition}[\foo]Let $A$ be a set of irreducible$^{\ref{fnumber}}$ polynomials in $\ZZ[x_1,\ldots,x_k]$, and $\Gamma$ a finite set in $\RR^k$ which contains points of nullification for elements in $A$. The polynomial part of the modified Lazard projection $\ProjBM(A)$ is the union of these sets of polynomials in $\ZZ[x_1,\ldots,x_{k-1}]$.
\begin{enumerate}
    \item All leading coefficients of the elements of $A$.
    \item Those trailing coefficients of the elements of $A$ for which $\Test$ returns $\bottom$.
    \item All discriminants of the elements of $A$.
    \item All resultants of pairs of distinct elements of $A$.
\end{enumerate}
The ``points'' part is the projection of the points previously identified together with those identified by $\Test$.
\begin{align}\label{eq:ProjBM}
	\ProjBM(A,\Gamma)_{poly}&:=\lc(A)\cup\bigcup_{\Test(f)=\bottom}\tc(f)\cup\disc(A)\cup\res(A),\\ 
	\ProjBM(A,\Gamma)_{points}&:=\{(\gamma_1,\ldots,\gamma_{k-1})|(\gamma_1,\ldots,\gamma_k)\in\Gamma\}\label{eq:ProjBM2}\\ \nonumber
	&\hspace{2.5cm}\cup\bigcup_{\Test(f)\ne\bottom}\{\Test(f)|f\in A\}.
\end{align}
        Though not explicit in \cite{BrownMcCallum2020a}, any contents of the polynomials (with respect to $x_k$) also carry forward into $\ProjBM(A,\Gamma)$.
\end{definition}
$\LiftBM$ is modified from $L_L$ (Lazard lifting) by analogy with the modifications $L_B$ brings to $L_M$: at each stage we add to the CAD $D_k$ the points in $\Gamma_k$. This is possible because, from the first part of (\ref{eq:ProjBM2}), the projections of these points are already point cells in $D_{k-1}$.

\par

\section{A Single Equational Constraint}
\def\foo{\cite[Definition~2]{Englandetal2019a}}
\begin{definition}[\foo]
        An equational constraint (EC) is a polynomial equation logically
        implied by a quantifer-free formula $\Phi$. If it is an atom of the formula, it is said to
        be {explicit}; if not, then it is {implicit}. If the
        constraint is visibly an equality one from the formula, i.e. the
	formula $\Phi$ is $(f=0)\land\Phi'$, we say the constraint
        is {syntactically explicit}.
\end{definition}
$\Phi$ might have more than one equational constraint, but in this section we fix one of these to be the chosen equational constraint.
\subsection{In McCallum's CAD}
In \cite{McCallum1999a}, McCallum, inspired by \cite{Collins1998}, adapted the method of \cite{McCallum1985b} to lift an order-invariant cylindrical algebraic decomposition of $\RR^{n-1}$ to a sign-invariant cylindrical algebraic decomposition, not of $\RR^n$, but rather of the variety $V_f$ of a primitive\footnote{If $f$ is not primitive with respect to $x_n$, then when the content vanishes, we have no constraint on $x_n$. See \cite{DavenportEngland2016a} for more on contents.} polynomial $f\in\ZZ[x_1,\ldots,x_n]$ of positive degree in $x_n$. This inherited the nullification problem of \cite{McCallum1985b}, but is useful when the polynomial problem takes the form $f=0\land \Phi$, i.e. had an \emph{equational constraint} (see also \cite{Englandetal2019a}). The projection operator from $\RR^n$ to $\RR^{n-1}$ is much smaller: if $E$ is an irreducible basis for $f$, then 
\begin{equation}\label{eq:PM2}
P_M^{E}(A) = P_M(E)\cup\{\res_{x_n}(f,g):f\in E,g\in A\setminus E\}.
\end{equation}
If we compare this with the original $P_M(A)$ from (\ref{eq:PM}) we see that 
we have no need to compute $\disc(A\setminus E)$ or $\res(A\setminus E)$, which reduces the asymptotic complexity --- the double exponent on the number of polynomials drops from $n$ to $n-1$ \cite{Englandetal2019a}. 
\subsection{In the Lex-least Case}
An obvious question is whether we can do the same with the Brown--McCallum projection\footnote{This section is based on \cite[Chapter 4]{Nair2021b}, which used the Lazard projection instead.}. Let $E$ be the set of irreducible factors of the equational constraint.
As in (\ref{eq:PM2}), we can define a new projection operator from $\RR^n$ to $\RR^{n-1}$. 
\begin{definition}\label{def:modifiedBM}
	Let $E$ be an irreducible basis for the chosen equational constraint $f$. Define the modified Brown--McCallum projection $P_{BM}^E$ by
\begin{align}\label{eq:PBM2}
        P_{BM}^E(A,\Gamma)_{poly}&:=P_{BM}(E,\Gamma)_{poly}\cup\{\res(f,g): f\in E, g \in A\setminus E\}. \\
        P_{BM}^E(A,\Gamma)_{points}&:=P_{BM}(E,\Gamma)_{points}.
\end{align}
\end{definition}
Note that, as well as not computing any projection polynomials from $A\setminus E$ alone, we are also only interested in the feet of point curtains from $E$, not from the whole of $A$.

The following theorem is the lex-least equivalent of \cite[Theorem 2.2]{McCallum1999a}, and is proved in \cite{Nair2021b}.
\def\foo{\cite[Theorem 18]{Nair2021b}}
\begin{theorem}[\foo]\label{thm:signinvariant}
	Let $n\geq 2$ and let $f,\,g\in\RR[x_1,\ldots,x_n]$ be of positive
	degrees in the main variable $x_n$. Let $S$ be a connected subset of
	$\RR^{n-1}$. Suppose that $f$ is Lazard delineable on
	$S$, that $\res_{x_n}(f,g)$ is
	\lexl invariant on $S$, and that $V_f$ does not have a curtain on $S$. Then $g$
	is sign-invariant in each Lazard section of $f$ over~$S$.
\end{theorem} 

\def\foo{\cite[Theorem 25]{Nair2021b}}
\begin{theorem}[cf \foo]\label{thm:BM_main_thm}
	Let $A=\{f_1,\allowbreak\ldots,f_m\}$ be a set of pairwise relatively prime irreducible polynomials in $n$ variables $x_1,...,x_n$ of positive degrees in $x_n$, where $n\geq 2$. Let $\Gamma$ be a finite set of points in $\RR^n$. Let S be a subset of $\RR^{n-1}$ obtained via the Brown--McCallum algorithm, so that each element of $P_{BM}^E(A,\Gamma)_{poly}$ is \lexl invariant in $S$. Then every element of $A$ is \lexl invariant in the Lazard sections and sectors of every other element.
\end{theorem}
\begin{proof}
Set $f=f_1\ldots f_m$. Then the leading coefficient of $f$ is non-zero and hence \lexl  invariant in $S$. We know $\disc_{x_n}(f)$ can be expanded as follows:
         \begin{equation}\label{eqn:Lazard_main_proof}
                \disc_{x_n}(f_1\ldots f_m)={\Big(}\prod_{i=1}^m\disc_{x_n}(f_i){\Big)(}\prod_{1\leq i<j\leq m}\res_{x_n}(f_i,f_j){\Big)}.
         \end{equation}
       	 If everything in the RHS of (\ref{eqn:Lazard_main_proof}) is \lexl invariant in $S$, then $\disc_{x_n}(f)$ is \lexl invariant. 
	 Hence, by Theorem \ref{thm:projectiondelineable}, $f$ is Lazard delineable over $S$. Since  $\res_{x_n}(f_i,f_j)$ for all $1\leq i<j\leq m$ are \lexl invariant in $S$, the Lazard sections of any two $f_i$ and $f_j$ $(i\neq j)$ are either disjoint or the same. Hence every element of $A$ is \lexl invariant in the Lazard sections of every other element.
\end{proof}

\def\foo{\cite[Theorem 26]{Nair2021b}}
\begin{theorem}[\foo]\label{thm:FirstEC_BM}
	Let $A$ be a set of pairwise relatively prime irreducible polynomials in $n$ variables $x_1,...,x_n$ of positive degrees in $x_n$, where $n\geq 2$. Let $E$ be a subset of $A$ and let $\Gamma$ be a finite set of points in $\RR^n$. Let $S$ be a connected subset of $\mathbb{R}^{n-1}$. Suppose that each element of $P_{BM}^E(A,\Gamma)_{poly}$ is \lexl invariant in $S$. Then each element of $E$ is Lazard delineable on $S$ and exactly one of the following must be true.
\begin{enumerate}
        \item The hypersurface $V_E$ defined by the product of the elements of $E$ has a curtain over~$S$.
        \item  The Lazard sections over $S$ of the elements of $E$ are pairwise disjoint. Each element of $E$ is \lexl invariant in every such Lazard section. Each element of $A \setminus E$ is sign invariant in every such Lazard section.
\end{enumerate}
	Hence, off the curtains of $E$, the lifted cells from the $\RR^{n-1}$ \lexl-invariant decomposition for $P_{BM}^E(A,\Gamma)$ form a sign-invariant decomposition of $V_E$ for $A$.
\end{theorem}

\begin{proof}
If there are no points in $S$ such that $E$ is nullified, then Theorem \ref{thm:projectiondelineable} implies that each element of $E$ is Lazard delineable since $P_{BM}(E,\Gamma)_{poly}\subseteq P_{BM}^E(A,\Gamma)_{poly}$. Otherwise $V_E$ has a curtain on $S$ and the Lazard sections of the elements of $E$ that do not contain curtains are pairwise disjoint. Now let $f\in E$ (such that $f$ does not have a curtain over $S$), let $\sigma$ be a section of $f$ over $S$ 
	and let $g\in A$. If $g\in E$, then by Theorem \ref{thm:projectiondelineable}, $g$ is \lexl invariant in every section and sector of $f$. If $g\notin E$ then, $R=\res(f,g)\in P_{BM}^E(A,\Gamma)_{poly}$ and by the assumption $R$ is \lexl invariant in $S$. Then by Theorem \ref{thm:signinvariant}, $g$ is sign invariant in $\sigma$.
\end{proof}
The question of what happens on the curtains of $E$ is taken up in the next section.

\section{Curtains and A Single Equational Constraint}
McCallum lifts order invariance to sign invariance in \cite{McCallum1999a} and order invariance in  \cite{McCallum2001},  
but only if the input polynomials do not have any curtains. Since curtains are already an issue for the base algorithm, other problems caused by curtains for equational constraints went unnoticed.
\subsection{Classifying Curtains}
In the case of a single equational constraint some curtains can be dealt simply. Therefore, we classify curtains further.

\begin{definition}\label{def:pointcurtain}
	Let $f\in\RR[x_1,\ldots, x_n]$ and $\alpha\in\RR^{n-1}$. We say that $f$ has a point curtain at $\alpha$ if $f(\alpha,y)=0$ for all $y\in\RR$, and there exists a Euclidean open neighbourhood $U\subset\RR$ of $\alpha$ such there exists no $\beta\in U\backslash\{\alpha\}$ such that $f(\beta,y)=0$ for all $y\in\RR$. 
\end{definition}

\begin{example} We provide one example to illustrate the difference.\label{eg:Point_nonpoint}
	\begin{itemize}
		\item Point Curtain: $f(x,y,z)=x^2+zy^2-z$ has point curtains at $(0,1)$ and $(0,-1)$. 
		\item Non-Point Curtain: Consider $f(x,y,z)=-x^3y^3z - xy^4z + xy^3z^2 + x^4 + 2x^3z + x^2y - x^2z + 2xyz - 2xz^2$. It has a non-point curtain on  $(0,y)$ for all $y\in\RR$. \end{itemize}
\end{example}
	
Note that in $\RR^3$, all non-point curtains are explicit curtains.	The example for non-point curtains in Example~\ref{eg:Point_nonpoint} is of the form $x\cdot$(quintic) so it is an explicit curtain (see Definition~\ref{def:implicit_explicit}). The curtains in Example~\ref{eg:EC_IMC} are explicit non-point, and implicit point curtains, respectively. These curtains are plotted in Figure~\ref{fig:three graphs}.

Now let us 
explain how curtains cause a problem when we try to exploit the hypersurfaces defined by an equational constraint.

\begin{example}\label{ex:nonpointcurtain}
	Let $f=x^2+y^2-1$ (which we assume to be an EC), $g_1=z-x-1$ and
	$g_2=z-y-1$. Then
	$\res_z(f,g_1)=x^2+y^2-1=\res_z(f,g_2)$, and this gives us no information
	about $\res(g_1,g_2)$. In such cases, when the EC has a non-point curtain, it
	becomes impossible to use $P_{BM}^E$ to detect the intersections of the other
	constraints on that curtain. In this example, $f$ does not have positive degree in $z$, so would not be considered anyway.
\end{example}

\begin{example}\label{ex:pointcurtain}
	Let $f=x-yz$ (which we assume to be an EC), $g_1=z-x$ and
	$g_2=z-y$ (which we assume are not ECs). 
	Then
	$\res(f,g_1)=yx-x$ and $\res(f,g_2)=y^2-x$ and   $\res(\res(f,g_1),\res(f,g_2))=x^2(1-x)$.  This gives us  information
	about the interaction of $g_1$ and $g_2$. In such cases, when the EC has a point curtain, $P_{BM}^E$ detects the intersections of the other
	constraints on that curtain. 
\end{example}	

Being able to classify curtains is just as important. For a given point $\beta\in\RR^{n-1}$, where $f\in\RR[x_1,\dots,x_n]$ nullifies, one can classify a curtain with $\beta$ in its foot by looking at the \lexl semi-valuation. If $\nu'_\beta(f) \not=0$ one checks the 1-cell neighbours. If any one of those cells also has non-zero semi-valuation at this point, then the curtain is a non-point curtain. 
\subsection{Point Curtains}\label{sec:PointCurtains}
Point curtains do not pose a problem for the lifting process.

\def\foo{\cite[Proposition~6]{Nair2021b}}
\begin{proposition}[\foo]\label{cor:pointcurtaininteraction}
	Let $f,g\in\RR[x_1,\ldots,x_n]$ and $\alpha\in\RR^{n-1}$ and suppose that the variety $V_f$ defined by $f$ contains a curtain at $\alpha$. If $\res(f,g)(\alpha)=0$ then $V_g$ contains a curtain at $\alpha$ or intersects $V_f$ at finitely many points over $\alpha$.  	
\end{proposition}

The modified projection operator $P_{BM}^E$ only projects information about the equational constraint and the resultants between the equational and non-equational constraints. Test $T$ in $P_{BM}^E$ will detect curtains with point feet in $\RR^{n-1}$, and will return $\bottom$ if it thinks there are non-point curtains. Any curtains, point or non-point, detected by $P_{BM}$ at lower levels lift to curtains with non-point feet in $\RR^{n-1}$. The following theorem shows that $P_{BM}^E$ actually produces a sign-invariant CAD for curtains with point feet in $\RR^{n-1}$.

\def\foo{\cite[Theorem~20]{Nair2021b}}
\begin{theorem}[\foo]\label{thm:point_curtain_exempt}
	Let $f,g\in\RR[x_1,\ldots,x_n]$, $E$ be an irreducible basis for $f$ our equational constraint, and suppose that $\alpha\in\RR^{n-1}$ is a cell of $P_{BM}^E(\{f,g\},\Gamma)_{points}$. 
	Then $g$ is sign invariant in the sections and sectors of $f_\alpha g_\alpha$. 
\end{theorem}

\begin{proof}
	By assumption $V_f$ has a point curtain at $\alpha$. This implies that there exists a neighbourhood $B$ of $\alpha$ in $\RR^{n-1}$ such that  $\nu_\beta^\prime (f) =0^{n-1}$ for $\beta \in B\setminus\{\alpha\}$ and $\nu_\alpha^\prime (f)\neq 0^{n-1}$. Because we are using the projection operator $P_{BM}(E)$, where $E$ is the set of irreducible factors of $f$, the polynomial $f$ is  \lexl invariant with respect to the CAD produced by the projection operator $P_{BM}(E)$. Since the polynomial $f$ is \lexl invariant with respect to the CAD and $V_f$ has a curtain at $\alpha$,  
	the CAD consists of cells of the form $\alpha\times ( a,b)$ or point cells $\alpha\times\{c\}$, where $a,b,c\in\RR$.

	We know $\res(f,g)(\alpha)=0$. Then by Proposition \ref{cor:pointcurtaininteraction}, $V_g$ has a point curtain at $\alpha$ or intersects $V_f$ at finitely many points. If $V_g$ has a curtain at $\alpha$ then $g$ is sign invariant (zero) in all cells over $\alpha$. Suppose $V_f$ only intersects $V_g$ at finitely many points and let $Z=\{\beta_1,\ldots,\beta_k\}$ (arranged in increasing order) be the roots of $f_\alpha g_\alpha$. Then the cells above $\alpha$ will be of the form $\alpha\times (a,b)$ or point cells $\alpha\times\{c\}$, where $a,b,c\in Z$. Note that $\beta_i$ has to be either a root of $f_\alpha$ or $g_\alpha$ for each $i$. Since $g=0$ intersects $f=0$ at only finitely many points and all roots of $g_\alpha$ are in $Z$, $g$ is non-zero in $\alpha\times(\beta_i,\beta_{i+1})$. $g$ is zero on the point cells $(\alpha,\beta_i)$ where $\beta_i$ is a root of $g$, and non-zero otherwise.
\end{proof}

Theorem \ref{thm:point_curtain_exempt} is the first point at which the merger of Lazard-style CAD and equational constraint theory is not straight-forward.

\subsection{Non-point Curtains}\label{sec:NPC}
Even if there are non-point curtains present, we know from Theorems \ref{thm:FirstEC_BM} and \ref{thm:point_curtain_exempt} that use of $P_{BM}^E$ at the first (projection) layer and $P_{BM}$ at subsequent layers generates a sign-invariant CAD away from the non-point curtains. So the challenge is purely confined to the non-point curtains (if any). Example \ref{ex:nonpointcurtain} shows that we need to consider the resultants of polynomials other than equational constraints, as $P_{BM}^E$ has lost too much information compared with $P_{BM}$. Hence our strategy is the following.
\begin{enumerate}
\item\label{stepOne} Project with $P_{BM}^E$ at the first layer and $P_{BM}$ at subsequent layers, and lift to $\RR^{n-1}$ to get decompositions $D_1,\ldots,D_{n-1}$, and to $D_n$, where the other constraints are sign-invariant over the non-curtain and point-curtain cells where $f=0$.
\item If there are no non-point curtains we are done. Otherwise let $(\alpha_1^{(i)},\ldots,\alpha_{n-1}^{(i)})$ be the sample points of the cells $C_i\in D_{n-1}$ forming the non-point curtains. We know the chosen equational constraints is zero on these cells.
\item Project all other constraints to $\RR^1$ using $P_{BM}$. Let $P_i$ be the polynomial part of the projection at level $\RR^i$.
\item\label{StepNP} In $\RR^1$, consider the cells whose sample points are $\alpha_1^{(i)}$, and decompose them further according to $P_1$, refining $D_1$ to $\widehat D_1$, introducing the necessary sample points.
\item In $\RR^2$, consider those cells whose sample points are $(\alpha_1^{(i)},\alpha_2^{(i)})$.  Their base in $\RR^1$ has been split in the previous stage, but they may also need to be split further in the $x_2$ direction by the polynomials in $P_2$. Appropriate new sample points will need to be computed.
\item\label{StepLast}Continue to $\RR^{n-1}$ and then lift these cells to $\RR^n$.
  \item\label{StepNB} We have subdivided some cells in each $D_i$ to get $\widehat D_i$. If we want a completely cylindrical decomposition, these subdivisions may need to be extended to the whole of each $\widehat D_i$. 
\end{enumerate}

Step \eqref{StepNB} is purely a data structure manipulation: no new polynomial computations need to be done beyond the computation of sample points for the newly-split cells.

 To decompose a curtain of an equational constraint is the same as decomposing a set $S\times \RR$ where $S\subset\RR^{n-1}$. When working in a curtain, the equational constraint does not give us any information on how the rest of the polynomial constraints interact. 
The following theorem 
validates our method of decomposing curtains separately. 

\def\foo{\cite[Theorem 21]{Nair2021b}}
\begin{theorem}[\foo]\label{thm:NonPoint}
	Let $A=\{f,g_1,\ldots g_m\}$ be a set of polynomial constraints  where $f$ is the equational constraint. Let $D$ be the CAD obtained by the applying the steps \eqref{stepOne}--\eqref{StepNB}. Then $g_1,\ldots, g_m$ are  sign-invariant in the cells of $D$ where $f=0$. 
\end{theorem}
\begin{proof}
	Let $S$ be a connected subset of $\RR^{n-1}$ such that $f$ is Lazard delineable over $S$ and the set $\{\res_{x_n}(f,g_j) | 1\le j\le m\}$ 
	 is  \lexl invariant over $S$. If $V_f$ does not contain a non-point curtain on $S$, then Theorems~\ref{thm:FirstEC_BM} and \ref{thm:point_curtain_exempt} imply the result. Let us assume the contrary. Since $V_f$ contains a non-point curtain on $S$, the aforementioned steps splits up $S$ into connected, non-intersecting subsets, where $f$ is Lazard delineable and $P_{BM}(\{g_1,\ldots,g_m\},\Gamma)$ is  \lexl invariant \cite{BrownMcCallum2020a}. Hence by Theorem \ref{thm:signinvariant}, $g_1,\ldots,g_m$ are  \lexl invariant in the cells of the curtains of $V_f$, hence also sign-invariant. 
\end{proof}

\section{Multiple Equational Constraints}

We now generalise the previous ideas to include multiple equational constraints, which involves modifying several levels of projection.

\subsection{Higher-dimensional Curtains}

\def\foo{\textsc{cf\cite[Definition 43]{Nair2021b}}}
\begin{definition}[\foo]\label{def:curtain}
	A variety $C\subseteq\RR^n$ is called an $m$-curtain if, whenever
	$(x_1,\ldots,x_n)\in C$, then  $(x_1,\ldots,x_{n-m},\allowbreak y_{n-m+1},\ldots,y_n)\in C$ for all
	$(y_{n-m+1},\ldots,y_n)\in\RR^m$.
\end{definition}

In other words, $C$ is an $m$-curtain if it is a union of fibres of $\RR^n \to \RR^{n-m}$.
\def\foo{\textsc{cf\cite[Definition 44]{Nair2021b}}}
\begin{definition}[\foo]\label{def:hascurtain}
	Suppose $S\subseteq\RR^{n-m}$ and $f\in\RR[x_1,\ldots,x_n]$. We say
	that $V_f$ (or $f$) has a $m$-curtain at $S$ if for all $(\alpha_1,\ldots,\alpha_{n-m})\in S$ and
	for all $(y_{n-m+1},\ldots,y_n)\in\RR^m$ we have $f(\alpha_1,\ldots,\alpha_{n-m},y_{n-m+1},\ldots,y_n)=0$. We call $S$ the foot of the curtain.
\end{definition}

\subsection{The Modified Projection}
For multiple equational constraints, we modify our projection operator further, inspired by \cite{McCallum2001}. We now assume that we are given a quantifier-free formula of the form

\begin{equation*}
	(f_1=0)\wedge\ldots\wedge(f_k=0)\wedge \Phi,
\end{equation*}
where $\Phi$ contains polynomial constraints $g_j$.

In this formula, the equational constraints are given by the polynomials $f_1,\ldots,f_k$. Since a solution set to this Tarski formula must satisfy the equational constraints, we focus our attention on the intersection of equational constraints. 
\begin{definition}\label{def:MultipleEC_PL}
	Let $A$ be a finite set of irreducible polynomials in $\RR[x_1,\ldots,x_n]$ with $n\geq 2$ and let $E$ be a subset of $A$. The modified Brown--McCallum projection operator for multiple equational constraints $\hat{P}_{BM}^E(A)$ is the subset of $\RR[x_1,\ldots,x_{n-1}]$ consisting of the following polynomials. 
	\begin{enumerate}
		\item All leading coefficients of the elements of $A$.
		\item The trailing coefficients of the elements of $A$ required by test $\Test$.
		\item All discriminants of the elements of $A$.\label{item:disc}
		\item All resultants $\res_{x_n}(f,g)$, where $f\in E$ and $g\in A$ and $f\neq g$.  
	\end{enumerate}
\end{definition}
This differs from Definition \ref{def:modifiedBM} precisely as \cite{McCallum2001} differs from \cite{McCallum1999a} --- by adding item (\ref{item:disc}).
As before, we shall choose equational constraints and set $E$ as the set of their irreducible factors. 

The way we take advantage of multiple equational constraints is by using this operator recursively. The resultants between the chosen equational constraint and the remaining equational constraints are equational constraints of the projected polynomials. Hence we take the equational constraints one at a time, so if we have $k$ equational constraints we will use $\hat{P}_{BM}^E(A)$ for the first $k$ projections: see \cite{Englandetal2019a} for more on multiple equational constraints. The following theorems validate the use of $\hat{P}_{BM}^E(A)$ recursively.

\def\foo{\cite[Theorem 22]{Nair2021b}}
\begin{theorem}[cf \foo]\label{thm:ordertolexl}
	Suppose $f,\,g\in \RR[x_1,\allowbreak\ldots,x_n]$ both have positive
	degree in $x_n$. Let $S\subset\RR^{n-1}$ be a
	connected subset such that $f$ does not have a curtain over $S$. If
	$\disc_{x_n}(g)$, $\ldcf_{x_n}(g)$ and $\res_{x_n}(f,g)$ are all \lexl invariant
	over $S$, then $g$ is \lexl invariant on every section and sector of $f$ over
	$S$.
\end{theorem}
\begin{proof}
	From Theorem \ref{thm:signinvariant} we know that $g$ is sign invariant in the sections of $f$. Since $\disc_{x_n}(g)$ and $\ldcf_{x_n}(g)$ are \lexl invariant over $S$, $g$ is Lazard delineable over $S$ from Theorem \ref{thm:projectiondelineable}. Hence $g$ is \lexl invariant in the sections and sectors of~$f$. 
\end{proof}

\subsection{Correctness off Curtains}
\def\foo{\cite[Theorem 27]{Nair2021b}}
\begin{theorem}[cf \foo]\label{thm:manyEC}
	Let $A$ be a set of irreducible polynomials in $x_1,\ldots,x_n$ of positive degree in $x_n$, where $n\ge2$. Let $E$ be a subset of $A$ and let $\Gamma$ be a finite set of points in $\RR^n$. Let $S$ be a connected subset of $\RR^{n-1}$. Suppose that each element of $\hat{P}_{BM}^E(A,\Gamma)$ is \lexl invariant in $S$. Then each element of $E$ is Lazard delineable over $S$ and exactly one of the following is true.
\begin{enumerate}
\item The hypersurface defined by the product of the elements of $E$ has a curtain whose foot is $S$. 
\item The Lazard sections of the elements of $E$ are pairwise disjoint over $S$. Each element of $A$ is \lexl invariant in such Lazard sections.
\end{enumerate}
\end{theorem}
\begin{proof}
	If we are in case 1, then at least one element of $E$ is zero over the whole of the curtain, so case 2 is meaningless. Hence the two cases are disjoint.
\par
So assume that $S$ is not a foot of a curtain for any element of $E$. Then, since  $\hat{P}_{BM}^E(A,\Gamma)\supset P_{BM}(E)$, the Lazard sections of the elements of $E$ are pairwise disjoint over $S$.
Then by Theorem \ref{thm:ordertolexl} the elements of $A$ are \lexl invariant in such Lazard sections.  
\end{proof}

Hence, if we never encounter curtains, we have a \lexl invariant stack over $S$. Assuming we have (after computing resultants to propagate equational constraints) an equational constraint in each of $x_n,\ldots,x_{n-k+1}$, then $k$ applications of $\hat{P}_{BM}^E$ followed by $n-k-1$ applications of ${P}_{BM}$ and finally $n-1$ applications of Lazard lifting, will give us a \lexl invariant (for $A$) CAD  of $V_{f_1,\ldots,f_k}$ in the absence of curtains.
\subsection{Curtains}\label{sec:recurse}
But what if we encounter a curtain of an equational constraint (either an original equation constraint, or a resultant of equational constraints) in the lifting?  Here we have to take a leaf out of \S\ref{sec:NPC}, where we projected the whole of $A\setminus E$ over the curtain. There may, however, be several curtains, possibly at different levels of lifting, and these curtains may be of derived equational constraints.

Hence we suggest that, \emph{if} curtains are detected in the lifting phase, we follow the methodology of \S\ref{sec:NPC} but do a complete projection of $A$ to $\RR^1$ using $P_{BM}$, and then follow steps \eqref{StepNP}--\eqref{StepLast} to refine the decomposition over those stacks that have curtains in them.

\section{Examples}\label{sec:Examples}
More examples are given in \cite[\S7.2]{Tonks2021a} and \cite{Davenportetal2023b}. The example we consider here is the first example in  \cite[\S7.2]{Tonks2021a}, taken from \cite[\S2.2.3]{Dolzmannetal2004a}.
As originally posed, the problem is
\begin{equation}\label{eq:EX1Orig}
\exists u\,\exists v~-uv+x=0\land -uv^2+y=0\land -u^2+z=0.
\end{equation}
The order chosen by the heuristics described in \cite{Tonks2021a} is  $u\prec v\prec y \prec z \prec x$.
Because we have multiple equational constraints, we can use Gr\"obner bases, as suggested in \cite{Wilsonetal2012a,DavenportEngland2017a,Englandetal2019a}. This converts the Cylindrical Algebraic Decomposition problem into
\begin{equation}\label{eq:EX1GB}
\{vx-y,uv-x,\mathbf{ux-vz}\},
\end{equation}
where the emboldened equation is taken as the equational constraint.
This projects as follows:
\begin{equation}\label{eq:EX1Proj}
\{-v^2z+uy,\mathbf{u^2-z}\} ;
\{\mathbf{uv^2-y}\} ;
\{v\} ;
\{u\}.
\end{equation}
Having chosen $uz-vz$ as the initial equational constraint, the curtains are $u=0\land v=0$ and $u=0\land z=0$, but in fact they are decomposed into 15 cells (each into 9 but three are duplicates) as given in \cite[Table 7.1]{Tonks2021a}. As described in \S\ref{sec:recurse}, we have to lift $A\setminus E$ i.e. $\{vx-y,uv-x\}$ over these curtains.
\begin{equation}\label{eq:EX1ProjRest}
\{vx-y, uv-x\} ;
\emptyset, ;
\{y,uv^2-y\} ;
\{v\} ;
\{u\}
\end{equation}
The combined result is a CAD with 591 cells.
\par
If we merely adopt the ``Single Equational Constraint'' methodology, the Lazard projection (as on \cite[p. 238]{Tonks2021a}) is
\begin{equation}\label{eq:EX1Proj1EC}
\{-v^2z+uy,u^2-z\} ;
\{y,uv^2-y\} ;
\{v\} ;
\{u\}
\end{equation}
which gives a total of 951 cells, due to the extra $y$ polynomial. 
\par
This $y$ is a trailing coefficient which the Brown--McCallum projection does not need, so in this case the behaviour would be the same as multiple equation constraints (which is to be expected, as the savings after the first projection only manifest when there are more than two polynomials in all).

\section{Complexity Analysis}
\begin{definition}
	Let $\{f_i\}$ be a finite set of polynomials. The combined maximum degree 
	of $\{f_i\}$ is the maximum element of the 
	set
	\begin{equation*}
	\bigcup_j{\Big\{} \deg_{x_j}\prod_i  f_i {\Big\}}.
	\end{equation*}	
	
\end{definition}

\def\foo{\cite[Section 6.1]{McCallum1985b}}
\begin{definition}[\foo]\label{def:md_poly}
	A set of polynomials has the $(m,d)$-property if it can be partitioned into $m$ sets, such that the combined maximum degree of each set is less than or equal to d.
\end{definition}

The outcome of the complexity analysis carried out in \cite[Chapter 8]{Nair2021b} (based on \cite{McCallum1985b}) is summarised in Table \ref{tbl:complexity_summary}. If the set of inputs has the $(m,d)$-property then the set of projected polynomials after projection has the $(M,2d^2)$-property, with $M$ as in Table \ref{tbl:complexity_summary}.

\renewcommand{\arraystretch}{1.5}
\begin{table}[ht]
	\vskip-9pt  
	\caption{Growth of polynomials in CAD} 
	\label{tbl:complexity_summary} 
	\makebox[.5\textwidth]{\begin{tabular}{|c | c | c | c | c|}
	\hline 
	Theory by && Original & Single EC & Multiple EC\\
	\hline 
	McCallum &\cite{McCallum1985b}& 	$\left\lfloor \frac{(m+1)^2}{2} \right\rfloor $ &  $\left\lfloor\frac{5m+4}{4}\right\rfloor$ & $\left\lfloor\frac{11m}{4}\right\rfloor$ \\
	\hline
	Lazard &\cite{McCallumetal2019a}& 	$\left\lfloor    \frac{(m+1)^2}{2}    \right\rfloor $ & $\left\lfloor    \frac{5m+3}{4}    \right\rfloor $ & $\left\lfloor    \frac{9m-1}{4}    \right\rfloor $ \\ 
	\hline
	Brown-McCallum &\cite{BrownMcCallum2020a}& $\frac{m(m+1)}{2}   $ & $\left\lfloor    \frac{5m+2}{4}    \right\rfloor  $ & $\left\lfloor    \frac{9m-2}{4}    \right\rfloor $ \\ 
	\hline
	\end{tabular}}
	\vskip-9pt  
\end{table}\renewcommand{\arraystretch}{1}

We note that there is very little improvement in the asymptotics as we move down the table: the only leading term to change is in the Multiple EC column, from 11 to 9. The differences in practice are more substantial: \cite[p 145]{BrownMcCallum2020a} shows that their method (``Improved Lazard'') only requires 27\% of the cells of the basic Lazard method.

This shows that, if we have only one equational constraint, the ``Single EC'' methods are better. However, if we have multiple equational constraints, the ``Multiple EC'' methods are better: for projections with two (or more) equational constraints, the number of polynomials is bounded by $O(m^4)$ for the original methods, $O(m^2)$ for the ``Single EC'' and $O(m)$ for the ``Multiple EC'' methods.

\section{Conclusion}
We have analysed both single and multiple equational constraints.  Table \ref{tbl:complexity_summary} shows that, if we only have one equation constraint, the ``single'' method is distinctly better at the cost of only producing a sign-invariant decomposition
, which is sufficient for all known applications.
\subsection{Single Equational Constraint}
Theorem \ref{thm:FirstEC_BM} tells us that we can transport the single equational work of \cite{McCallum1999a} (which can fail because of nullification) to the Brown--McCallum projection in the cases where  \cite{McCallum1999a} works. The complexity gain \cite{Englandetal2019a} is the same: the double exponent of the number of polynomials is reduced by 1. Theorem \ref{thm:point_curtain_exempt} extends this to point curtains, with the same complexity improvement.

Theorem \ref{thm:NonPoint} shows that we can work with an equational constraint even in the presence of non-point curtains, but the efficiency gain \emph{on the curtain} is more modest, and does not lend itself to a simple formulation.

\subsection{Multiple Equational Constraints}
Theorem \ref{thm:manyEC}  tells us that we can transport the single equational work of \cite{McCallum2001} (which fails with nullification) to the Brown--McCallum projection in the cases where  \cite{McCallum2001} worked.  The complexity gain \cite{Englandetal2019a} is the same: the double exponent of the number of polynomials is reduced by $k$, the number of consecutive layers where we can apply $\hat{P}_{BM}^E$.

In the presence of curtains, our method is probably more efficient than purely using \cite{BrownMcCallum2020a}, but this is hard to quantify. See \S\ref{sec:Examples}.
\subsection{Future Work}
There are several open questions.
\begin{enumerate}
  \item We have talked about ``the number of consecutive layers where we'' have an equational constraint. This is implicit in \cite{McCallum2001} as well. But it appears that, in the absence of curtains, it should be possible to ``mix and match'' between $PBM$ and $\hat{P}_{BM}^E$, using  $\hat{P}_{BM}^E$ at every layer where there is an equations constraint.  But formalising this would require a suitable induction hypothesis.
  \item \S\ref{sec:PointCurtains} deals more neatly with point curtains in the case of a single equational constraint. When this curtain occurs in the lift from $\RR^{n-1}$ to $\RR^n$ we can probably do the same for multiple equational constraints.  But what if the point curtain is elsewhere? We currently propose treating is as a non-point curtain --- can we do better?
    \item \cite{Englandetal2019a} suggested using Gr\"obner bases rather than resultants to propagate equational constraints, in the setting of \cite{McCallum2001}. That has been shown to work in \S\ref{sec:Examples}.
        \item
    \cite[see also \cite{Bradfordetal2016a}]{Bradfordetal2013b} introduced the concept of a ``truth table invariant CAD'', where equational constraints might only govern part of a formula, as in $((f_1=0)\land\Phi_1)\lor((f_2=0)\land\Phi_2)$. This theory build on \cite{McCallum1999a}, so should transfer.
    \end{enumerate}
In addition, \cite[\S9.1]{Nair2021b} asks whether \lexl is the only other valuation (besides sign in \cite{Collins1975} and order in \cite{McCallum1985b}) that we should consider. It is noted there that ``The property of upper-semicontinuity guarantees that if two polynomials are valuation invariant in a set,
then their product is valuation invariant in the same set and vice-versa, which helps
us take advantage of equational constraints''.
So what other upper-semicontinuous valuations might we consider?

\pagebreak
\bibliography{../../../../jhd.bib}
\end{document}